\documentclass[10pt]{article}
\usepackage{}

\usepackage{pgf}
\usepackage[square, comma, sort&compress, numbers]{natbib}
\usepackage[english]{babel}
\usepackage[utf8]{inputenc}
\usepackage{datetime}
\usepackage{geometry}
\usepackage{array}
\usepackage{xspace}

\usepackage{amsmath,amsfonts,amssymb,amsopn,amsthm,amscd}
\theoremstyle{plain}
\usepackage{graphicx} 
\usepackage{epstopdf}
\usepackage{enumerate}
\allowdisplaybreaks[4]

\def\Box{\vcenter{\vbox{\hrule\hbox{\vrule
     \vbox to 8.8pt{\hbox to 10pt{}\vfill}\vrule}\hrule}}}

\usepackage{indentfirst}
\newcommand{\Ff}{{\mathbb F}}

\newcommand\cc{{\mathcal C}}        %

\newcommand{\F}{ {{\mathbb F}} }

\newcommand\cF{{\mathbf c}}

\def\Tr{\operatorname{Tr}}

\newtheorem{thm}{Theorem}[section]
\newtheorem{lem}[thm]{Lemma}

\newtheorem{cor}[thm]{Corollary}

\newtheorem{remark}{Remark}
\def\Tr{\operatorname{Tr}}

\begin{document}

%

\title{Two constructions of asymptotically optimal codebooks via the trace functions
}

\author{Xia Wu$^1$, Wei Lu$^{1*}$, Xiwang Cao$^{2}$, Ming Chen$^3$
}

\maketitle

\let\thefootnote\relax\footnotetext{This work was supported by the National Natural Science Foundation of China (Grant No. 11801070, 11771007, 61572027) and the Basic Research
Foundation (Natural Science).}

\let\thefootnote\relax\footnotetext{$^1$School of Mathematics, Southeast University, Nanjing 210096, China. (Email: wuxia80@seu.edu.cn)}

\let\thefootnote\relax\footnotetext{$^1$School of Mathematics, Southeast University, Nanjing 210096, China. (Email: luwei1010@seu.edu.cn)}

\let\thefootnote\relax\footnotetext{$^2$Department of Math, Nanjing University of Aeronautics and Astronautics, Nanjing 211100, China. (Email: xwcao@nuaa.edu.cn)}

\let\thefootnote\relax\footnotetext{$^3$School of Information Science and Engineering, Southeast University, Nanjing 210096, China. (Email: chenming@seu.edu.cn)}

\let\thefootnote\relax\footnotetext{$^*$Corresponding author. (Email: luwei1010@seu.edu.cn)}

\begin{abstract}
In this paper, we present two new constructions of complex codebooks with multiplicative characters, additive characters and trace functions over finite fields, and determine the maximal cross-correlation amplitude of these codebooks. We prove that the codebooks we constructed are asymptotically optimal with respect to the Welch bound. Moreover, in the first construction, we generalize the result in \cite{NM5}. In the second construction, we generalize the results in \cite{NM1}, we can  achieve Welch bound for any odd prime $p$,  we also derive the whole distribution of  their inner products. The parameters of these codebooks are new.

{\bf Keywods}: Codebook, asymptotic optimality, Welch bound, trace function.

\end{abstract}

\section{Introduction}
An $(N, K)$ codebook $\cc=\{\mathbf{c}_0,\mathbf{c}_1, ...,\mathbf{c}_{N-1}\}$ is a set of $N$ unit-norm complex vectors  $\mathbf{c}_i \in\mathbb{C}^K$ over an alphabet \textit{A}, where $i=0, 1, \ldots, N-1$. The size of \textit{A} is called the alphabet size of $\cc$. As a performance measure of a codebook in practical applications, the maximum magnitude of inner products between a pair of distinct vectors in $\cc$ is defined by
$$
I_{max}(\cc)=\underset{0\leq i\neq j \leq N-1}{\max}|\mathbf{c}_i\mathbf{c}_j^H|,
$$
where $\mathbf{c}_j^H$ denotes the conjugate transpose of the complex vector $\mathbf{c}_j$.
To evaluate an $(N, K)$ codebook $\cc$, it is important to find the minimun achievable $I_{max}(\cc)$ or its lower bound. The Welch bound \cite{Welch} provides a well-known lower bound on $I_{max}(\cc)$,
$$
I_{max}(\cc)\geq I_W=\sqrt{\frac{N-K}{(N-1)K}}.
$$
The equality holds if and only if for all pairs of $(i,j)$ with $i\neq j$
$$|\mathbf{c}_i\mathbf{c}_j^H|=\sqrt{\frac{N-K}{(N-1)K}}.$$
%

A codebook $\cc$ achieving the Welch bound equality is called a maximum-Welch-bound-equality (MWBE) codebook \cite{MWEB} or an equiangular tight frame \cite{frame}. MWBE codebooks are employed in various applications including code-division multiple-access(CDMA) communication systems \cite{Apl.CDMA}, communications \cite{MWEB}, combinatorial designs \cite{Apl.cd1,Apl.cd2,Apl.cd3}, packing \cite{Apl.pa}, compressed sensing \cite{Apl.cs}, coding theory \cite{Apl.code} and quantum computing \cite{Apl.qc}. To our knowledge, only the following  MWBE codebooks are presented as follows:

\begin{itemize}
\item $(N, N)$ orthogonal MWBE codebooks for any $N>1$ \cite{MWEB}, \cite{Apl.cd3};
\item $(N,N-1)$ MWBE codebooks for $N>1$ based on discrete Fourier transformation matrices \cite{MWEB}, \cite{Apl.cd3} or $m$-sequences \cite{MWEB};
\item $(N, K)$ MWBE codebooks from conference matrices \cite{Apl.pa}, \cite{MWEB1}, where $N=2K=2^{d+1}$ for a positive integer $d$ or $N=2K=p^d+1$  for a prime $p$ and a positive integer $d$;
\item $(N, K)$ MWBE codebooks based on $(N, K, \lambda)$ difference sets in cyclic groups \cite{Apl.cd3} and abelian groups \cite{Apl.cd1}, \cite{Apl.cd2};
\item $(N, K)$ MWBE codebooks from $(2, k, \nu)$-Steiner systems \cite{MWEB2};
\item $(N, K)$ MWBE codebooks depended on graph theory and finite geometries \cite{F1,F2,F3,F4}.
\end{itemize}

The construction of an MWBE codebook is known to be very hard in general, and the known classes of MWBE codebooks only exist for very restrictive $N$ and $K$. Many researches have been done instead to construct near optimal codebooks, i.e., codebook $\cc$ whose $I_{max}(\cc)$ nearly achieves  the Welch bound. In \cite{MWEB}, Sarwate gave some nearly optimal codebooks from codes and signal sets.  As an extension of the optimal codebooks based on difference sets, various types of near optimal codebooks based on almost difference sets, relative difference sets and cyclotomic classes were proposed, see \cite{Apl.cd1,NM2,NM5,NM6,zhou}. Near optimal codebooks constructed from binary row selection sequences were presented in \cite{NM1,NM3Yu}. In \cite{AL2,Luo,Luo2,Luo3}, some near optimal codebooks were constructed  via Jacobi sums and hyper Eisenstein sum.

In this paper, we present two new constructions of complex codebooks with multiplicative characters, additive characters and trace functions over finite fields, and determine the maximal cross-correlation amplitude of these codebooks. The key ideas in our new constructions are based on transitivity of trace and the Fourier expansion of characters. We prove that the codebooks we constructed are asymptotically optimal with respect to the Welch bound. Moreover, in the second construction, we also derive the whole distribution of  their inner products. As a comparison, in Table 1, we list the parameters of some known classes of near optimal codebooks and those of the new ones.

This paper is organized as follows. In section 2, we recall some notations and basic results which will be needed in our discussion. In section 3, we present our first construction of codebooks.  In section 4, we present our second construction of codebooks. In section 5, we conclude this paper.

 \begin{table}[!htbp]
 \newcommand{\tabincell}[2]{\begin{tabular}{@{}#1@{}}#2\end{tabular}}
 \caption{\small The parameters of codebooks asymptotically meeting the Welch bound}
 \label{tabl1}
 \centering
  \setlength{\tabcolsep}{1mm}{
 \begin{tabular}{|c|c|c|c|}
   \hline
   Parameters $(N,K)$ & $I_{max}$ & References \\ \hline
  \tabincell{c}{ $(p^n,K=\frac{p-1}{2p}(p^n+p^{n/2})+1)$ \\
  with odd $p$ }& $\frac{(p+1)p^{n/2}}{2pK}$ & \cite{NM1} \\ \hline
  $(q^2, \frac{(q-1)^2}{2})$, $q=p^s$ with odd $p$ & $\frac{q+1}{(q-1)^2}$ & \cite{NM5} \\ \hline
  $q(q+4), \frac{q+1}{2}$, $q$ is a prime power & $\frac{1}{q+1}$ & \cite{Li} \\ \hline
    $q, \frac{(q+3)(q+1)}{2}$, $q$ is a prime power  & $\frac{\sqrt{q}+1}{q-1}$ & \cite{Li} \\ \hline
   $(p^n-1,\frac{p^n-1}{2})$ with odd $p$ & $\frac{\sqrt{p^n}+1}{p^n-1}$ & \cite{NM3Yu} \\ \hline
   $(q^l+q^{l-1}-1,q^{l-1})$ for any $l>2$ & $\frac{1}{\sqrt{q^{l-1}}}$ & \cite{zhou} \\ \hline
    \tabincell{c} {$((q-1)^k+q^{k-1},q^{k-1}$),\\ for any $k>2$ and $q\geq4$} & $\frac{\sqrt{q^{k+1}}}{(q-1)^k+(-1)^{k+1}}$ & \cite{AL2} \\ \hline
  \tabincell{c} {$((q-1)^k+K,K$), for any $k>2$, \\ where $K=\frac{(q-1)^k+(-1)^{k+1}}{q}$ } & $\frac{\sqrt{q^{k-1}}}{K}$ & \cite{AL2} \\ \hline
   \tabincell{c}{$((q^s-1)^n+K,K$, \\ for any $s>1$ and $n>1$, \\ where $K=\frac{(q^s-1)^n+(-1)^{n+1}}{q})$ } & $\frac{\sqrt{q^{sn+1}}}{(q^s-1)^n+(-1)^{n+1}}$ & \cite{Luo} \\ \hline
   \tabincell{c}{$((q^s-1)^n+q^{sn-1},q^{sn-1})$, \\ for any $s>1$ and $n>1$} & $\frac{\sqrt{q^{sn+1}}}{(q^s-1)^n+(-1)^{n+1}}$ & \cite{Luo} \\ \hline
   \tabincell{c}{$(q-1,\frac{q(r-1)}{2r})$, \\ $r=p^t, q=r^s$, with odd $p$ and $p\nmid s$} & $\frac{\sqrt{r}}{\sqrt{q}(\sqrt{r}-1)K}$ & this paper \\ \hline
   \tabincell{c}{$(q^2,\frac{q(q+1)(r-1)}{2r})$, \\ $r=p^t, q=r^s$, with odd $p$} & $\frac{(r+1)q}{2rK}$ & this paper \\ \hline
 \end{tabular}}
 \end{table}

\section{Preliminaries}
In this section, we introduce some basic results on trace functions, characters and character sums over finite fields, which will be needed in the following sections.

\subsection{Trace functions of finite fields}
Let $\Ff_p$, $\Ff_r$ and $\Ff_q$ be finite fields, with $\Ff_p\subseteq\Ff_r\subseteq\Ff_q$. Let $\Tr_{q/r}(\cdot)$ be the trace functions from $\Ff_q$ to $\Ff_r$. with
$$\Tr_{q/r}(x)=x+x^r+x^{r^2}+...+x^{\frac{q}{r}},\ x\in \Ff_q.$$

Transitivity of trace in \cite{field} is stated as
$$\Tr_{r/p}\cdot\Tr_{q/r}(x)=\Tr_{q/p}(x),\ x\in \Ff_q,$$
it plays important role in our constructions.
%

\subsection{characters over finite fields}
Let $\Ff_q$ be a finite field, in this subsection, we recall the definitions of the additive and multiplicative characters of $\Ff_q$.

For each $a\in \mathbb{F}_q$, an additive character of $\mathbb{F}_q$ is defined by the function $\chi_a(x)=\zeta_p^{\Tr_{q/p}(ax)}$, where $\zeta_p$ is a primitive $p-$th root of complex unity. By the definition, $\chi_a(x)=\chi_1(ax)$. If $a=0$, we call $\chi_0$ the trivial additive character of $\mathbb{F}_q$. Let $\widehat{\Ff_q}$ be the set of all additive characters of $\Ff_q$.The orthogonal relation of additive characters (see \cite{field}) is given by
$$
 \sum_{x\in \mathbb{F}_q}\chi_a(x)=\left\{
            \begin{array}{ll}
              q,& \hbox{if\ $a=0$,}\\
              0,& \hbox{otherwise.}
            \end{array}
          \right.
$$

As in \cite{field}, the multiplicative character of $\mathbb{F}_q$ is defined as follows. For $j=0,1,...,q-2$, the functions $\varphi_j$ defined by
$$\varphi_j(\alpha^i)=\zeta_{q-1}^{ij},$$
are all the multiplicative characters of $\mathbb{F}_q$,
where $\alpha$ is a primitive element of $\mathbb{F}_q^*$, and $0\leq i\leq q-2$. If $j=0$, we have $\varphi_0(x)=1$ for any $x\in \mathbb{F}_q^*$, $\varphi_0$ is called the trivial multiplicative character of $\mathbb{F}_q$. Let $\widehat{\Ff_q^*}$ be the set of all the multiplicative characters of $\Ff_q^*$.

Let $\varphi$ be a multiplicative character of $\Ff_q$, the orthogonal relation of multiplicative characters (see \cite{field}) is given by
$$
 \sum_{x\in \mathbb{F}_q^*}\varphi(x)=\left\{
            \begin{array}{ll}
              q-1,& \hbox{if\ $\varphi=\varphi_0$,}\\
              0,& \hbox{otherwise.}
            \end{array}
          \right.
$$

\subsection{Character sums over finite fields}
Let $\varphi$ be a multiplicative character of $\Ff_q$ and $\chi$  an additive character of $\Ff_q$. Then the Gauss sum over $\Ff_q$ is given by
\begin{equation*}
  G(\varphi,\chi):=G_q(\varphi,\chi)=\sum_{x\in \mathbb{F}_q^*}\varphi(x){\chi}(x).
\end{equation*}
It is easy to see the absolute value of $G_q(\varphi,\chi)$ is at most $q-1$, but is much smaller in general, the following lemma shows all the cases.
\begin{lem}{\rm \cite[Theorem 5.11]{field}}\label{gauss}
Let $\varphi$ be a multiplicative character and $\chi$ an additive character of $\Ff_q$. Then the Gauss sum $G_q(\varphi,\chi)$ satisfies
$$
 G_q(\varphi,\chi)=\left\{
            \begin{array}{ll}
              q-1, & \hbox{if\ $\varphi=\varphi_0$, $\chi=\chi_0$,} \\
              -1,& \hbox{if\ $\varphi=\varphi_0$, $\chi\neq\chi_0$,}\\
              0,& \hbox{if\ $\varphi\neq\varphi_0$, $\chi=\chi_0$.}
            \end{array}
          \right.
$$
For $\varphi\neq\varphi_0$ and $\chi\neq\chi_0$, we have $\left|G_q(\varphi,\chi)\right|=\sqrt{q}$.
\end{lem}

\begin{lem}{\rm \cite{field}}\label{gauss1}
Gauss sums for the finite field $\Ff_q$ satisfy the following properties:

(1)$G_q(\varphi,\chi_{ab})=\overline{\varphi(a)}G_q(\varphi,\chi_{b})$ for $a\in \Ff_q^*$, $b\in \Ff_q$, where $\overline{\varphi(a)}$ denotes the complex conjugate of $\varphi(a)$;

(2)$G_q(\varphi,\chi)G_q(\overline{\varphi},\chi)=\varphi(-1)q$,

\end{lem}

\begin{lem}{\rm \cite{field}}\label{gauss2}
Let $\varphi$ be a multiplicative character of $\Ff_q$, then for any $c\in \Ff_q^*$, we have
$$
\varphi(c)=\frac{1}{q}\sum_{\chi}G_q(\varphi,\overline{\chi})\chi(c),
$$
where the sum is extended over all additive characters $\chi$ of $\Ff_q$.

\end{lem}

This is the Fourier expansion of the multiplicative character $\varphi$ in terms  of the additive characters of $\Ff_q,$ with Gaussian sums appearing as Fourier coefficients.

\subsection{A general construction of codebooks}
Let $D$ be a set and $\#D=K$. Let $E$ be a set of some functions which satisfy
$$f:D\rightarrow S,\ \ \hbox{where\  S\  is\  the\ unit\ circle.}$$
A general construction of codebooks is stated as follows in the complex plane,
$$\cc(D;E)=\{\cc_f:=\frac{1}{\sqrt{K}}(f(x))_{x\in D}, f\in E\}.$$

\section{The first construction of codebooks}
In this section, by multiplicative characters, transitivity of trace and fourier expansion, we construct new series of codebooks  which nearly meet the Welch bound.

We first recall the construction in \cite{NM5}. Let

$$
D=\{x\in \Ff_{q}^*: \eta(x+1)=-1\},
$$
where $\eta$ is the quadratic multiplicative character of $\Ff_q$, $\eta(0)$ is defined as 0 for convenience. Let $\#D=K$ and $E=\widehat{\Ff_q^*}$.

The codebooks $C(D;E)=\{\cF_{\varphi}=\frac{1}{\sqrt{K}}(\varphi(x))_{x\in D}: \varphi\in E\}$, which are constructed in \cite{NM5} are asymptotically optimal when $q\rightarrow \infty$.

In this section, we generalize this construction in \cite{NM5} by the transitivity of trace functions.

Let $p$ be an odd prime, $r=p^t$  and $q=r^s$, where $p\nmid s$.  Let $\eta$ be the quadratic multiplicative character of $\Ff_{r}$. Let

$$
D=\{x\in \Ff_{q}^*: \eta(\Tr_{q/r}(x+1))=-\eta(s)\},
$$
and $\#D=K$. Let $E=\widehat{\Ff_q^*}$.

We define a codeword of length $K$ as
$$
\cF_{\varphi}=\frac{1}{\sqrt{K}}(\varphi(x))_{x\in D},\ \varphi\in E
$$
and construct the following $(N,K)$ codebook $\cc(D;E)$ as:
\begin{equation}\label{con2}
\cc(D;E)=\{\cF_{\varphi}=\frac{1}{\sqrt{K}}(\varphi(x))_{x\in D}: \varphi\in E\}.
\end{equation}

We set
$$
 \delta_1(x)=\left\{
            \begin{array}{ll}
              \frac{1-\eta(s)\eta(\Tr_{q/r}(x+1))}{2},& \hbox{if\ $\Tr_{q/r}(x+1)\neq$ 0,}\\
              0,& \hbox{if\ $\Tr_{q/r}(x+1)=$ 0.}
            \end{array}
          \right.
$$
Through the definition of $D$, we know that

$$
 \delta_1(x)=\left\{
            \begin{array}{ll}
              1, & \hbox{if\ $x\in D$,} \\
              0,& \hbox{otherwise.}
            \end{array}
          \right.
$$

\begin{thm}\label{th51}
With the above notation, we have $N=q-1$, $K=\frac{q(r-1)}{2r}$ and
$$
I_{max}(\cc(D;E))\leq\frac{\sqrt{r}}{\sqrt{q}(\sqrt{r}-1)}<1.
$$

\end{thm}
\begin{proof}
By the definition of $D$ and $E$, it is easy to see that $N=q-1$ and $K=\frac{q(r-1)}{2r}$.

For any characters $\varphi_i$ and $\varphi_j$ in $\widehat{\Ff_q^*}$, where $1\leq i \neq j\leq q-2$. Let $\varphi:=\varphi_i\overline{\varphi_j}$, since $i\neq j$, $\varphi$ is nontrival. Then we have

\begin{eqnarray*}
&&K(\cF_{\varphi_i}\cF_{\varphi_j}^H)\\
&=&\sum_{x\in D}\varphi_i(x)\overline{\varphi_j(x)}\\
&=&\sum_{x\in D}\varphi(x)\\
&=&\sum_{x\in \Ff_q^*}\varphi(x)\delta_1(x)\\
&=&\sum_{x\in \Ff_q^*,\atop\Tr_{q/r}(x+1)\neq 0}\varphi(x)\frac{1-\eta(s)\eta(Tr_{q/r}(x+1))}{2} \\
&=& \sum_{x\in \Ff_q^*}\varphi(x)\frac{1-\eta(s)\eta(Tr_{q/r}(x+1))}{2}\\
&&-\sum_{x\in \Ff_q^*,\atop\Tr_{q/r}(x+1)= 0}\varphi(x)\frac{1-\eta(s)\eta(Tr_{q/r}(x+1))}{2}\\
&=&\frac{1}{2}\sum_{x\in \Ff_q^*}\varphi(x)\eta(Tr_{q/r}(x+1))-\frac{1}{2}\sum_{x\in \Ff_q^*,\atop\Tr_{q/r}(x+1)= 0}\varphi(x)\\
&=&-\frac{1}{2}\eta(s)A-\frac{1}{2}B,
\end{eqnarray*}
where $A$ and $B$ are defined in Lemma \ref{lem52} and Lemma \ref{lem53}.

By the results in Lemma \ref{lem52} and Lemma \ref{lem53}, we get

\begin{eqnarray*}
I_{max}(\cc(D;E)) &=& max\{|\cF_{\varphi_i}\cF_{\varphi_j}^H|:1\leq i \neq j\leq q-2\}\\
&\leq&\frac{|A|+|B|}{2K}\leq \frac{\sqrt{r}}{\sqrt{q}(\sqrt{r}-1)}<1.
\end{eqnarray*}

\end{proof}

Using Theorem \ref{th51}, we can derive the ratio of $I_{max}(\cc(D;E))$ of the proposed codebooks to that of the MWBE codebooks and show the near-optimality of the proposed codebooks as in the following theorem.

\begin{thm}\label{th52}
Let $I_W$ be the magnitude of the optimal correlation, i.e. the Welch bound equality, for the given $N$, $K$ in the current section. We have

$$\lim_{r\rightarrow\infty}\frac{I_{max}(\cc(D;E))}{I_W}=1,$$
then the codebooks we proposed are near optimal.

\end{thm}
\begin{proof}
Note that $N=q-1$ and $K=\frac{q(r-1)}{2r}$. Then the corresponding Welch bound is
$$
I_W=\sqrt{\frac{N-K}{(N-1)K}}=\sqrt{\frac{qr-2r+q}{q(q-2)(r-1)}}.
$$
It can easily be seen that
\begin{eqnarray*}
1\leq\frac{I_{max}(\cc(D;E))}{I_W}&\leq&\frac{\sqrt{r}}{\sqrt{q}(\sqrt{r}-1)}\frac{\sqrt{q(q-2)(r-1)}}{\sqrt{qr-2r+q}}\\
&<&\frac{r\sqrt{q-2}}{(\sqrt{r}-1)\sqrt{qr-2r+q}}\\
&=&\frac{\sqrt{1-\frac{2}{q}}}{(1-\frac{1}{\sqrt{r}})\sqrt{1-\frac{2}{q}+\frac{1}{r}}}.
\end{eqnarray*}

Since $r=p^t$ and $q=r^s$, it follows that $\lim_{r\rightarrow+\infty }\frac{I_{max}(\cc(D;E))}{I_W}=1$. The codebook $\cc(D;E)$ asymptotically meets the Welch bound. This completes the proof.
\end{proof}
\begin{remark}
When $r=q$, it is exactly the first construction of codebooks in \cite{NM5}.
\end{remark}

The Lemmas  which are used in the proof of Theorem \ref{th51} are as follows.

Let $\varphi$ be a  multiplicative character of $\Ff_{q}$ and $\varphi^*$  the restriction of $\varphi$ to $\F_{r}$. Let $\chi_a$ be an  additive character of $\Ff_{q}$, and $\chi_a^*$  the restriction of $\chi_a$ to $\F_{r}$, where $a\in \Ff_q$.  Let $\psi_b$ be an  additive character of $\Ff_{r}$, where $b\in \Ff_r$. Let  $\chi=\chi_1$ and $\psi=\psi_1$ be the canonical additive characters of $\Ff_{q}$ and $\Ff_{r}$ respectively.

\begin{lem}\label{lem51}
$\chi^*$  is nontrivial if and only if $p\nmid s$.
\end{lem}
\begin{proof}
For any $b\in \Ff_r$, we have
\begin{eqnarray*}
\chi^*(b)&=&\chi(b)=\zeta_p^{\Tr_{q/p}b}=\zeta_p^{\Tr_{r/p}\Tr_{q/r}b}=\zeta_p^{\Tr_{r/p}b\Tr_{q/r}1}\\
&=&\zeta_p^{\Tr_{r/p}(bs)}=\psi(bs)=\psi_s(b).
\end{eqnarray*}
It follows that $\chi^*=\psi_s$, the result is obtained.
\end{proof}

\begin{lem}\label{lem52}
Let $A=\sum_{x\in \Ff_q^*}\varphi(x)\eta(Tr_{q/r}(x+1))$, then we have
$$
|A|=\left\{
            \begin{array}{ll}
              \sqrt{q},& \hbox{if\ $\eta\overline{\varphi}^*$ is nontrivial,}\\
              \frac{\sqrt{q}}{\sqrt{r}},& \hbox{if\ $\eta\overline{\varphi}^*$ is trivial.}
            \end{array}
          \right.
$$
\end{lem}

\begin{proof}
By the Fourier expansion in Lemma \ref{gauss2}, A can be rewritten as
\begin{eqnarray*}
A&=&\sum_{x\in \Ff_q^*}\varphi(x)\frac{1}{r}\sum_{b\in \Ff_r}G_r(\eta,\overline{\psi_b})\psi_b(\Tr_{q/r}(x+1))\\
&=&\frac{1}{r}\sum_{b\in \Ff_r}G_r(\eta,\overline{\psi_b})\sum_{x\in \Ff_q^*}\varphi(x)\psi(\Tr_{q/r}(bx+b))\\
&=&\frac{1}{r}\sum_{b\in \Ff_r}G_r(\eta,\overline{\psi_b})\sum_{x\in \Ff_q^*}\varphi(x)\chi(bx+b)\\
&=&\frac{1}{r}\sum_{b\in \Ff_r}G_r(\eta,\overline{\psi_b})\sum_{x\in \Ff_q^*}\varphi(x)\chi(bx)\chi(b)\\
&=&\frac{1}{r}\sum_{b\in \Ff_r}G_r(\eta,\overline{\psi_b})G_q(\varphi,\chi_b)\chi(b)\\
&=&\frac{1}{r}\sum_{b\in \Ff_r^*}\overline{\eta(-b)}G_r(\eta,\overline{\psi})\overline{\varphi(b)} G_q(\varphi,\chi)\chi(b)\\
&=&\frac{1}{r}\eta(-1)G_r(\eta,\overline{\psi})G_q(\varphi,\chi)\sum_{b\in \Ff_r^*}\eta(b)\overline{\varphi(b)}\chi(b)\\
&=&\frac{1}{r}\eta(-1)G_r(\eta,\overline{\psi})G_q(\varphi,\chi)G_r(\eta\overline{\varphi}^*,\chi^*).\\
\end{eqnarray*}

Since $p\nmid s$, the result follows from Lemma \ref{gauss} and Lemma \ref{lem51}.
\end{proof}
We set $$
 \delta_2(x)=\frac{1}{r}\sum_{b\in \Ff_r}\psi_b(\Tr_{q/r}(x+1))
$$

It is easy to see that
$$\delta_2(x)=
          \left\{
            \begin{array}{ll}
              1,& \hbox{if\ $\Tr_{q/r}(x+1)=$ 0,}\\
              0,& \hbox{if\ $\Tr_{q/r}(x+1)\neq$ 0.}
            \end{array}
          \right.$$

\begin{lem}\label{lem53}
Let $B=\sum_{x\in \Ff_q^*,\atop\Tr_{q/r}(x+1)= 0}\varphi(x)$, then we have
$$
|B|=\left\{
            \begin{array}{ll}
              \frac{\sqrt{q}}{\sqrt{r}},& \hbox{if\ $\overline{\varphi}^*$ is nontrivial,}\\
              \frac{\sqrt{q}}{r},& \hbox{if\ $\overline{\varphi}^*$ is trivial.}
            \end{array}
          \right.
$$
\end{lem}
\begin{proof}
By the definition of $\delta_2(x)$, we have
\begin{eqnarray*}
B&=&\sum_{x\in \Ff_q^*,\atop\Tr_{q/r}(x+1)= 0}\varphi(x)=\sum_{x\in \Ff_q^*}\varphi(x)\delta_2(x)\\
&=&\sum_{x\in \Ff_q^*}\varphi(x)\frac{1}{r}\sum_{b\in \Ff_r}\psi_b(\Tr_{q/r}(x+1))\\
&=&\frac{1}{r}\sum_{b\in \Ff_r}\sum_{x\in \Ff_q^*}\varphi(x)\psi(\Tr_{q/r}(bx+b))\\
&=&\frac{1}{r}\sum_{b\in \Ff_r}\sum_{x\in \Ff_q^*}\varphi(x)\chi(bx+b)\\
&=&\frac{1}{r}\sum_{b\in \Ff_r}\sum_{x\in \Ff_q^*}\varphi(x)\chi(bx)\chi(b)\\
&=&\frac{1}{r}\sum_{b\in \Ff_r}G_q(\varphi,\chi_b)\chi(b)\\
&=&\frac{1}{r}\sum_{b\in \Ff_r^*}\overline{\varphi(b)}G_q(\varphi,\chi)\chi(b)\\
&=&\frac{1}{r}G_q(\varphi,\chi)G_r(\overline{\varphi}^*,\chi^*).\\
\end{eqnarray*}
Since $p\nmid s$, the result follows from Lemma \ref{gauss} and Lemma \ref{lem51}.
\end{proof}

 In Table \ref{table3}, we provide some explicit values of the parameters of the codebooks we proposed for some given $r$ and $q$, and corresponding numerical data of the Welch bound for comparison. The numerical results show  that the codebooks $\cc(D;E)$ asymptotically meet the Welch bound.

 \begin{table}[!htbp]
 \newcommand{\tabincell}[2]{\begin{tabular}{@{}#1@{}}#2\end{tabular}}
 \caption{\footnotesize Parameters of the $(N,K)$ codebook of Section III}
 \label{table3}
 \centering
 \setlength{\tabcolsep}{1mm}{
 \begin{tabular}{|c|c|c|c|c|c|c|}
   \hline
 $r$& $q$ & $N$ & $K$ & $I_{max}(\cc)$ & $I_W$ & $\frac{I_{max}(\cc)}{I_W}$ \\ \hline
 $3^2$&$3^4$& $80$& $36$& $0.1667$& $0.1244$& $1.3399$ \\ \hline
 $19$&$19^2$& $360$& $171$& $0.0683$& $0.0555$& $1.2310$ \\ \hline
 $179$&$179^2$& $32040$& $15931$& $0.006$& $0.0056$& $1.0748$ \\ \hline
 $3^5$&$3^{10}$& $59048$& $29403$& $0.0044$& $0.0041$& $1.0642$ \\ \hline
 $3^7$&$3^{14}$& $4782968$& $2390391$& $0.00046724$& $0.00045746$& $1.0214$ \\ \hline
 $5^3$&$5^6$& $244140625$& $121101500$& $6.5028e-05$& $6.541e-05$& $1.0080$ \\ \hline
 $7^3$&$7^{12}$& $1.19158e+20$& $9.5511e+19$& $7.2670e-11$& $7.2459e-11$& $1.0029$ \\ \hline
  $5^4$&$5^{16}$& $2.3283e+22$& $1.1623e+22$& $6.5746e-12$& $6.5641e-12$& $1.0016$ \\ \hline
  $19^5$&$19^{10}$& $6.1311e+12$& $3.0655e+12$& $4.0412e-07$& $4.0386e-07$& $1.0006$ \\ \hline
 \end{tabular}}
 \end{table}

\section{The second construction of codebooks}
In this section, by additive characters, transitivity of trace and fourier expansion, we construct new series of codebooks  which nearly meet the Welch bound.

In fact, we generalize the construction of codebooks in \cite{NM1} and show that  $I_{max}$ of the proposed codebooks asymptotically achieves the Welch bound. We also derive the whole distribution of inner products of the proposed codebooks. In  \cite{NM1} the Welch bound can be asymptotically achieved only for sufficiently large prime $p$, while in our construction the  Welch bound can be asymptotically  achieved for any odd prime $p$.

 We first recall the construction in \cite{NM1}.

Let $p$ be an odd prime and $q$ a power of $p$. Let $\eta$ be the quadratic multiplicative character of $\Ff_p$.   The construction in \cite{NM1} is equivalent to the following one.

Let
$$
D=\{x\in \Ff_{q^2}: \eta(\Tr_{q/p}x^T)=-1\},
$$
and $\#D=K$, where $T=q+1$.
Let $E=\widehat{\Ff_{q^2}}$.

The codebooks
$$C(D;E)=\{\cF_a=\frac{1}{\sqrt{K}}(\chi_a(x))_{x\in D}: \chi_a\in E\},$$
 are asymptotically optimal when $p\rightarrow \infty$.

In this paper we generalize the construction in \cite{NM1}.

Let $p$ be an odd prime, $r=p^t$  and $q=r^s$.  Let $\eta$ be the quadratic multiplicative character of $\Ff_{r}$. Let $\chi_a$ be an  additive character of $\Ff_{q^2}$, where $a\in \Ff_{q^2}$.  Let  $\chi=\chi_1$  be the canonical additive characters of $\Ff_{q^2}$.

Let
$$
D=\{x\in \Ff_{q^2}: \eta(\Tr_{q/r}x^T)=-1\},
$$
and $\#D=K$, where $T=q+1$.

Let $E=\widehat{\Ff_{q^2}}=\{\chi_a: a\in \Ff_{q^2}\}$.

We define a codeword of length $K$ as
$$
\cF_{a}=\frac{1}{\sqrt{K}}(\chi_a(x))_{x\in D},
$$
and construct the following $(N,K)$ codebooks as:
\begin{equation}\label{con2}
\cc(D;E)=\{\cF_{a}: a\in \Ff_{q^2}\}.
\end{equation}
By the definition of $D$ and $E$, it is easy to see that $N=q^2$ and $K=\frac{q(q+1)(r-1)}{2r}$.

%
%

The following theorem gives the whole distribution of inner products of  $\cc(D;E)$.

\begin{thm}\label{th52}
Let $C_{a_i}C_{a_j}^H$ be the inner product between a pair of code vectors $C_{a_i}$ and $C_{a_j}$ of the proposed codebooks $\cc(D;E)$,
where $a_i$ and $a_j$ are distinct elements in $\Ff_{q^2}$, $0\leq i\neq j\leq q^2-1$. Then $C_{a_i}C_{a_j}^H$ has the following distribution
$$
 C_{a_i}C_{a_j}^H=\left\{
            \begin{array}{ll}
              \frac{1}{K}\frac{r-1}{2r}q, & \hbox{$\frac{r+1}{2r}(q^4)-\frac{r-1}{2r}(q^3)-q^2$\ times,}\\
              \frac{1}{K}\frac{-r-1}{2r}q, &\hbox{$\frac{r-1}{2r}(q^4+q^3)$\ times.}
            \end{array}
          \right.
$$

\end{thm}

\begin{proof}

We set
$$
 \delta_3(x)=\left\{
            \begin{array}{ll}
              \frac{1-\eta(\Tr_{q/r}x^T)}{2},& \hbox{if\ $\Tr_{q/r}x^T\neq$ 0,}\\
              0,& \hbox{if\ $\Tr_{q/r}x^T=$ 0.}
            \end{array}
          \right.
$$
By the definition of $D$, we known that

$$
 \delta_3(x)=\left\{
            \begin{array}{ll}
              1, & \hbox{if\ $x\in D$,} \\
              0,& \hbox{otherwise.}
            \end{array}
          \right.
$$

As the notations defined in this section, we have
\begin{eqnarray*}
K(C_{a_i}C_{a_j}^H)&=&\sum_{x\in D}\chi_{a_i}(x)\overline{\chi_{a_j}}(x)\\
&=&\sum_{x\in D}\chi((a_i-a_j)x)\\
&=&\sum_{x\in D}\chi(ax), \ 0\neq a\in \Ff_{q^2}\\
&=&\sum_{x\in \Ff_{q^2}}\chi(ax)\delta_3(x)\\
&=&\sum_{x\in \Ff_{q^2} \atop \Tr_{q/r}x^T\neq 0}\chi(ax)\frac{1-\eta(\Tr_{q/r}x^T)}{2}\\
&=&\sum_{x\in \Ff_{q^2}}\chi(ax)\frac{1-\eta(\Tr_{q/r}x^T)}{2}\\
&&-\sum_{x\in \Ff_{q^2}\atop\Tr_{q/r}x^T= 0}\chi(ax)\frac{1-\eta(\Tr_{q/r}x^T)}{2}\\
&=&-\frac{1}{2}\sum_{x\in \Ff_{q^2}}\chi(ax)\eta(\Tr_{q/r}x^T)\\
&&-\frac{1}{2}\sum_{x\in \Ff_{q^2}\atop\Tr_{q/r}x^T= 0}\chi(ax),\\
&=&-\frac{1}{2}Q-\frac{1}{2}P,
\end{eqnarray*}
where $Q$ and $P$ are defined in Lemma \ref{lem55} and Lemma \ref{lem56}.

By the results in Lemma \ref{lem55} and Lemma \ref{lem56}, we have

$$
 K(C_{a_i}C_{a_j}^H)=\left\{
            \begin{array}{ll}
              \frac{r-1}{2r}q, & \hbox{if\ $\Tr_{q/r}a^T=0$,} \\
              \frac{r-1}{2r}q, & \hbox{if\ $\Tr_{q/r}a^T=1$,}\\
              \frac{-r-1}{2r}q, & \hbox{if\ $\Tr_{q/r}a^T=-1$.}
            \end{array}
          \right.
$$

To find the distribution, we need the following illustrations.
$(1)$ For any $a\in \Ff_{q^2}^*$, there are $N$ groups $(a_i,a_j)$ satisfying $a_i-a_j=a$. \\
$(2)$ For any $z\in \Ff_q^*$, there are $T$ different $a\in \Ff_{q^2}^*$ which satisfy $z=a^T$.\\
$(3)$ There are $\frac{q}{r}-1$ different $z\in \Ff_q^*$  with $\Tr_{q/r}z=0$, and there are $\frac{q}{r}$ different $z\in \Ff_q^*$  with $\Tr_{q/r}z\neq 0$.\\
$(4)$ The number of squares and nonsquares in $\Ff_r^*$ is $\frac{r-1}{2}$.

Then it is easy to get the distribution of $C_{a_i}C_{a_j}^H$ as follows:
$$
 C_{a_i}C_{a_j}^H=\left\{
            \begin{array}{ll}
              \frac{1}{K}\frac{r-1}{2r}q, & \hbox{$\frac{r+1}{2r}(q^4)-\frac{r-1}{2r}(q^3)-q^2$\ times,}\\
              \frac{1}{K}\frac{r+1}{2r}q, &\hbox{$\frac{r-1}{2r}(q^4+q^3)$\ times.}
            \end{array}
          \right.
$$
This completes the proof.
\end{proof}

The following corollary  which is a direct consequence of Theorem \ref{th52} gives the upper bound of $I_{max}{\cc(D;E)}$.
\begin{cor}\label{cor51}
The maximum magnitude $I_{max}{\cc(D;E)}$ of inner products between a pair of code vectors of the proposed codebooks is upper bounded by
$$I_{max}{\cc(D;E)}=\underset{0\leq i\neq j \leq q^2-1}{\max}|\mathbf{c}_i\mathbf{c}_j^H|\leq\frac{1}{K}\frac{r+1}{2r}q.$$
\end{cor}

\begin{thm}\label{th53}
Let $I_W$ be the magnitude of the optimal correlation, i.e. the Welch bound equality, for the given $N$, $K$. We have

$$\lim_{r\rightarrow\infty}\frac{I_{max}(\cc(D;E))}{I_W}=1,$$
then the codebooks we proposed are near optimal.

\end{thm}
\begin{proof}
Note that $N=q^2$ and $K=\frac{q(q+1)(r-1)}{2r}$. Then the corresponding Welch bound is
$$
I_W=\sqrt{\frac{N-K}{(N-1)K}}=\sqrt{\frac{2rq-(q+1)(r-1)}{(q+1)^2(q-1)(r-1)}}.
$$
It can easily be seen that
\begin{eqnarray*}
1\leq\frac{I_{max}(\cc(D))}{I_W}&\leq& \frac{1}{K}\frac{r+1}{2r}q\sqrt{\frac{2rq-(q+1)(r-1)}{(q+1)^2(q-1)(r-1)}}\\
&=&\sqrt{\frac{(r+1)^2(q-1)}{(r-1)[2rq-(q+1)(r-1)]}}\\
&=&\sqrt{\frac{(r+1)^2(q-1)}{(r-1)[q(r+1)-(r-1)]}}\\
&<&\sqrt{\frac{(r+1)^2(q-1)}{(r-1)[q(r+1)-(r+1)]}}\\
&=&\sqrt{\frac{r+1}{r-1}}.
\end{eqnarray*}
It follows that $\lim_{r\rightarrow+\infty }\frac{I_{max}(\cc(D;E))}{I_W}=1$. The codebooks $\cc(D;E)$ asymptotically meet the Welch bound. This completes the proof.
\end{proof}

\begin{remark}
When $r=p$, the construction we proposed is exactly the one in \cite{NM1}. In \cite{NM1} the codebooks are construced by selecting $K$ rows from $N\times N$ Hadamard matrices based on binary row selection sequence, but we use the additive character of finite field and a set $D$ to construct our codebooks. It is pointed out in \cite{NM1} that $I_{max}$ of the codebooks asymptotically achieves the Welch bound equality only for sufficiently large prime number $p$. For small $p$, $I_{max}$ of the codebooks has larger value than the Welch bound equality, though $q\rightarrow+\infty$. It is hard to find sufficiently large primes. While, in our construction, for any fixed odd prime $p$, we can let $t$ increase and then $r=p^t\rightarrow+\infty$. So $I_{max}(\cc(D))$  can  asymptotically achieve the Welch bound equality.
\end{remark}

The Lemmas  which are used in the proof of Theorem \ref{th52} are as follows.
\begin{lem}\label{54}
Let $f_b(x)=\Tr_{q/p}(bx^T)$, $\widehat{f_b}(a)=\sum_{x\in \Ff_{q^2}}\omega^{\Tr_{q/p}(bx^T)+\Tr_{q^2/p}(ax)}$, where $b\in \Ff_r$ and $\omega$ is the complex $p$-th root of unity. Then we have
$$\widehat{f_b}(a)=
          \left\{
            \begin{array}{ll}
              -q\omega^{-\Tr_{q/p}{\frac{a^T}{b}}},& \hbox{if\ $b\in \Ff_{r}^*,$}\\
              0,& \hbox{if\ $b=0$.}
            \end{array}
          \right.$$


where $a\in \Ff_{q^2}$.
\end{lem}
\begin{proof}
When $b\in \Ff_{r}^*$, it holds that
\begin{eqnarray*}
\widehat{f_b}(a)&=&\sum_{x\in \Ff_{q^2}}\omega^{\Tr_{q/p}(bx^T+ax+a^qx^q)}\\
&=&\sum_{x\in \Ff_{q^2}}\omega^{\Tr_{q/p}(b[x^q(x+\frac{a^q}{b})+\frac{a}{b}(x+\frac{a^q}{b})-\frac{a^{q+1}}{b^2}])}\\
&=&\sum_{x\in \Ff_{q^2}}\omega^{\Tr_{q/p}(b[(x^q+\frac{a}{b})(x+\frac{a^q}{b})-\frac{a^{q+1}}{b^2}])}\\
&=&\sum_{x\in \Ff_{q^2}}\omega^{\Tr_{q/p}(b[(x^q+(\frac{a^q}{b})^q)(x+\frac{a^q}{b})-\frac{a^{q+1}}{b^2}])}\\
&=&\sum_{x\in \Ff_{q^2}}\omega^{\Tr_{q/p}(b(x+\frac{a^q}{b})^{q+1}-\frac{a^{q+1}}{b})}\\
&=&\sum_{y\in \Ff_{q^2}}\omega^{\Tr_{q/p}(by^T-\frac{a^T}{b})}\\
&=&\omega^{-\Tr_{q/p}\frac{a^T}{b}}\sum_{y\in \Ff_{q^2}}\omega^{\Tr_{q/p}(by^T)}\\
&=&\omega^{-\Tr_{q/p}\frac{a^T}{b}}\sum_{z\in \Ff_{q}}\sum_{y\in \Ff_{q^2} \atop y^T=z}\omega^{\Tr_{q/p}(bz)}\\
&=&\omega^{-\Tr_{q/p}\frac{a^T}{b}}\sum_{z\in \Ff_{q}}\omega^{\Tr_{q/p}(bz)}\sum_{y\in \Ff_{q^2} \atop y^T=z}1\\
&=&\omega^{-\Tr_{q/p}\frac{a^T}{b}}[1+(q+1)\sum_{z\in \Ff_{q^*}}\omega^{\Tr_{q/p}(bz)}]\\
&=&-q\omega^{-\Tr_{q/p}\frac{a^T}{b}}
\end{eqnarray*}
When $b=0$, it is easy to see that $\widehat{f_b}(a)=0$.
\end{proof}

\begin{remark}
By the result in Lemma \ref{54}, $f_b(x)$ is a regular bent function, see \cite{H}.
\end{remark}

Let $\psi_b$ be an  additive character of $\Ff_{r}$, where $b\in \Ff_r$. Let  $\chi=\chi_1$ and $\psi=\psi_1$ be the canonical additive characters of $\Ff_{q^2}$ and $\Ff_{r}$ respectively.
We set $$
 \delta_4(x)=\frac{1}{r}\sum_{b\in \Ff_r}\psi_b(\Tr_{q/r}x^T),
$$

then
$$\delta_4(x)=
          \left\{
            \begin{array}{ll}
              1,& \hbox{if\ $\Tr_{q/r}x^T=$ 0,}\\
              0,& \hbox{if\ $\Tr_{q/r}x^T\neq$ 0.}
            \end{array}
          \right.$$

\begin{lem}\label{lem55}
Let $P=\sum_{x\in \Ff_{q^2} \atop \Tr_{q/r}x^T=0}\chi(ax)$, where $a\neq 0$, then we have
$$P=
          \left\{
            \begin{array}{ll}
              -\frac{q}{r}(r-1),& \hbox{if\ $\Tr_{q/r}a^T=$ 0,}\\
              \frac{q}{r},& \hbox{if\ $\Tr_{q/r}a^T\neq$ 0.}
            \end{array}
          \right.
$$
\end{lem}

\begin{proof}
It follows that
\begin{eqnarray*}
P&=&\sum_{x\in \Ff_{q^2} \atop \Tr_{q/r}x^T=0}\chi(ax)\\
&=&\sum_{x\in \Ff_{q^2}}\chi(ax)\delta_4(x)\\
&=&\sum_{x\in \Ff_{q^2}}\chi(ax)\frac{1}{r}\sum_{b\in \Ff_r}\psi_b(\Tr_{q/r}x^T)\\
&=&\frac{1}{r}\sum_{x\in \Ff_{q^2}}\sum_{b\in \Ff_r}\omega^{\Tr_{r/p}(b\Tr_{q/r}x^T)}\omega^{\Tr_{q^2/p}(ax)}\\
&=&\frac{1}{r}\sum_{b\in \Ff_r}\sum_{x\in \Ff_{q^2}}\omega^{\Tr_{q/p}(bx^T)+\Tr_{q^2/p}(ax)}\\
&=&\frac{1}{r}\sum_{b\in \Ff_r}\widehat{f_b}(a)\\
&=&\frac{1}{r}\sum_{b\in \Ff_r^*}-q\omega^{-\Tr_{q/p}{\frac{a^T}{b}}}\\
&=&-\frac{q}{r}\sum_{c\in \Ff_r^*}\omega^{\Tr_{q/p}{ca^T}}\\
&=&-\frac{q}{r}\sum_{c\in \Ff_r^*}\omega^{\Tr_{r/p}c\Tr_{q/r}{a^T}}\\
&=&-\frac{q}{r}\sum_{c\in \Ff_r^*}\psi(c\Tr_{q/r}{a^T})
\end{eqnarray*}

The result follows from the fact that
$$\sum_{c\in \Ff_{r}}\psi(c\Tr_{q/r}{x^T})=
          \left\{
            \begin{array}{ll}
              r,& \hbox{if\ $\Tr_{q/r}x^T=$ 0,}\\
              0,& \hbox{if\ $\Tr_{q/r}x^T\neq$ 0.}
            \end{array}
          \right.
$$

\end{proof}

\begin{lem}\label{lem56}
Let $Q=\sum_{x\in \Ff_{q^2}}\psi(ax)\eta(\Tr_{q/r}x^T)$, where $a\neq 0$. Then we have
$$Q=
          \left\{
            \begin{array}{ll}
              0,& \hbox{if\ $\Tr_{q/r}a^T=$ 0,}\\
              -q\eta(-\Tr_{q/r}a^T),& \hbox{if\ $\Tr_{q/r}a^T\neq$ 0.}
            \end{array}
          \right.
$$
\end{lem}

\begin{proof}
It holds that
\begin{eqnarray*}
Q&=&\sum_{x\in \Ff_{q^2}}\chi(ax)\eta(\Tr_{q/r}x^T)\\
&=&\sum_{x\in \Ff_{q^2}}\chi(ax)\frac{1}{r}\sum_{b\in \Ff_r}G_r(\eta,\overline{\psi_b})\psi_b(\Tr_{q/r}x^T)\\
&=&\frac{1}{r}\sum_{b\in \Ff_r}G_r(\eta,\overline{\psi_b})\sum_{x\in \Ff_{q^2}}\chi(ax)\psi_b(\Tr_{q/r}x^T)\\
&=&\frac{1}{r}\sum_{b\in \Ff_r}G_r(\eta,\overline{\psi_b})\sum_{x\in \Ff_{q^2}}\omega^{\Tr_{q^2/p}(ax)}\omega^{\Tr_{q/p}(b\Tr_{q/r}(x^T))}\\
&=&\frac{1}{r}\sum_{b\in \Ff_r}G_r(\eta,\overline{\psi_b})\sum_{x\in \Ff_{q^2}}\omega^{\Tr_{q^2/p}(ax)+\Tr_{q/p}(bx^T)}\\
&=&\frac{1}{r}\sum_{b\in \Ff_r^*}G_r(\eta,\overline{\psi_b})\widehat{f_b}(a)\\
&=&\frac{1}{r}\sum_{b\in \Ff_r^*}G_r(\eta,\overline{\psi_b})(-q\omega^{-\Tr_{q/p}\frac{a^T}{b}})\\
&=&-\frac{q}{r}\sum_{b\in \Ff_r^*}G_r(\eta,\overline{\psi_b})\omega^{\Tr_{r/p}(-\frac{1}{b}\Tr_{q/r}a^T)}\\
&=&-\frac{q}{r}\sum_{b\in \Ff_r^*}\overline{\eta}(-b)G_r(\eta,{\psi})\psi(-\frac{1}{b}\Tr_{q/r}a^T)\\
&=&-\frac{q}{r}G_r(\eta,{\psi})\sum_{c\in \Ff_r^*}\eta(c)\psi(c\Tr_{q/r}a^T)\\
&=&-\frac{q}{r}G_r(\eta,{\psi})G_r(\eta,\psi_{\Tr_{q/r}a^T}).
\end{eqnarray*}
The result follows from the fact that
$$G_r(\eta,\psi_{\Tr_{q/r}a^T})=
             \left\{
                 \begin{array}{ll}
                  0,& \hbox{if\ $\Tr_{q/r}a^T=$ 0,}\\
                 \overline{\eta}(\Tr_{q/r}a^T)G_r(\eta,{\psi}),& \hbox{if\ $\Tr_{q/r}a^T\neq$ 0,}
                \end{array}
              \right.
$$
and
$$G_r(\eta,{\psi})G_r(\eta,{\psi})=\eta(-1)r.$$
\end{proof}

 In Table \ref{table4}, we provide some explicit values of the parameters of the codebooks we proposed for some given $r$ and $q$, and corresponding numerical data of the Welch bound for comparison. The numerical results show  that the codebooks $\cc(D;E)$ asymptotically meet the Welch bound.

 \begin{table}[!htbp]
 \newcommand{\tabincell}[2]{\begin{tabular}{@{}#1@{}}#2\end{tabular}}
 \caption{\footnotesize Parameters of the $(N,K)$ codebook of Section IV}
 \label{table4}
 \centering
 \setlength{\tabcolsep}{1mm}{
 \begin{tabular}{|c|c|c|c|c|c|c|}
   \hline
 $r$& $q$ & $N$ & $K$ & $I_{max}(\cc)$ & $I_W$ & $\frac{I_{max}(\cc)}{I_W}$ \\ \hline
 $3^2$&$3^4$& $6561$& $2952$& $0.0152$& $0.0137$& $1.1166$ \\ \hline
 $19$&$19^2$& $130321$& $61902$& $0.0031$& $0.0029$& $1.0539$ \\ \hline
 $5^2$&$5^6$& $244140625$& $11719500$& $6.9329e-05$& $6.6609e-05$& $1.0408$ \\ \hline
 $3^3$&$3^6$& $531441$& $256230$& $0.0015$& $0.0014$& $1.0377$ \\ \hline
 $5^3$&$5^{6}$& $244140625$& $121101500$& $6.5028e-05$& $6.410e-05$& $1.0080$ \\ \hline
 $3^5$&$3^{10}$& $3.4868e+09$& $1.7362e+09$& $1.7075e-05$& $1.7005e-05$& $1.0041$ \\ \hline
  $7^3$&$7^{12}$& $1.9158e+20$& $9.5511e+19$& $7.2670e-11$& $7.2459e-11$& $1.0029$ \\ \hline
  $3^6$&$3^{12}$& $2.8243e+11$& $1.4102e+11$& $1.8868e-06$& $1.8843e-06$& $1.0014$ \\ \hline
  $13^3$&$13^6$& $2.3298e+13$& $1.1644e+13$& $2.0736e-07$& $2.0727e-07$& $1.0005$ \\ \hline
 \end{tabular}}
 \end{table}


\section{Concluding remarks}
%
In this paper, we presented two new constructions of codebooks asymptotically achieve the Welch bounds with multiplicative characters, additive characters, trace functions and Fourier expansion of finite  fields. Our two constructions are the generalizations of  the constructions in \cite{NM5} and \cite{NM1}, respectively.  The parameters of our codebooks are  flexible and new.

One feature of our construction is the combination of the character sums and the transitivity of trace functions. Since character sums of finite fields are one of the main methods used in the constructions of near optimal codebooks, we believe that using trace functions we can generalize more known codebooks constructed by character sums.

In our second construction, one key point is that $f_b(x)$ in the Lemma IV.4 is a regular bent function. We wonder if there exist other (weakly) regular bent functions such that  the codebooks constructed by them can asymptotically achieve the Welch bound.

\begin{thebibliography}{1}



\bibitem{Apl.cs} E. Candes and M. Wakin, ``An introduction to compressive sampling,'' IEEE Signal Process, vol. 25, no. 2, pp. 21-30, 2008.


\bibitem{Apl.pa} J. Conway, R. Harding, and N. Sloane, ``Packing lines, planes, etc.: Packings in Grassmannian spaces,'' Exp. Math., vol. 5, no. 2, pp. 139-159, 1996.


\bibitem{Apl.cd1} C. Ding, ``Complex codebooks from combinatorial designs,'' IEEE Trans. Inform. Theory, vol. 52, no. 9, pp. 4229-4235, 2006.


\bibitem{Apl.cd2} C. Ding and T. Feng, ``A generic construction of complex codebooks meeting the Welch bound,''
IEEE Trans. Inf. Theory, vol. 53, no. 11, pp. 4245-4250, 2007.


\bibitem{Apl.code} P. Delsarte, J. Goethals, and J. Seidel, ``Spherical codes and designs,'' Geometriae Dedicate, vol. 67, no. 3, pp. 363-388, 1997.


\bibitem{MWEB2} M. Fickus, D. Mixon, and J. Tremain, ``Steiner equiangular tight frames,'' Linear Algebra Appl., vol. 436, no.5, pp. 1014-1027, 2012.


\bibitem{F1} M. Fickus and D. Mixon, ``Tables of the existence of equiangular tight frames,'' arXiv:1504.00253v2, 2016.

\bibitem{F2} M. Fickus, D. Mixon and J. Jasper, ``Equiangular tight frames from
hyperovals,'' IEEE Trans. Inf. Theory, vol. 62, no. 9, pp. 5225-5236, 2016.

\bibitem{F3} M. Fickus, J. Jasper, D. Mixon and J. Peterson, ``Tremain equiangular
tight frames,'' arXiv:1602.03490v1, 2016.

\bibitem{H} T. Helleseth, A. Kholosha, ``Monomial and quadratic bent functions over the
finite fields of odd characteristic,'' IEEE Trans. Inf. Theory, vol. 52, no. 5, pp. 2018-2032, 2006.

\bibitem{NM2} H. Hu and J. Wu, ``New constructions of codebooks nearly meeting the
Welch bound with equality,'' IEEE Trans. Inf. Theory, vol. 60, no. 2, pp. 1348-1355, 2014.


\bibitem{NM1} S. Hong, H. Park, T. Helleseth, and Y. Kim, ``Near optimal
partial Hadamard codebook construction using binary sequences
obtained from quadratic residue mapping,'' IEEE Trans. Inf. Theory, vol. 60, no. 6, pp. 3698-3705, 2014.

\bibitem{AL2} Z. Heng, C. Ding, Q. Yue, ``New constructions of asymptotically optimal
codebooks with multiplicative characters,'' IEEE Trans. Inf. Theory, vol. 63, no. 10, pp. 6179-6187, 2017.


\bibitem{frame} J. Kovacevic and A. Chebira, ``An introduction to frames,'' Found. Trends Signal Process., vol. 2, no. 1, pp. 1-94, 2008.

\bibitem{Li} C. J. Li, Y. Qin, Y. W. Huang,  ``Two families of nearly optimal codebooks,'' Des. Codes Cryptogr. vol. 75, no. 1, pp. 43-
57, 2015.

\bibitem{Luo} G. J. Luo, X. W. Cao, ``Two constructions of asymptotically optimal codebooks via the hyper Eisenstein sum,''  IEEE Trans. Inf. Theory, vol. 64, no. 10, pp. 6498-6505, 2018.

\bibitem{Luo2} G. J. Luo, X. W. Cao, ``New constructions of codebooks asymptotically achieving the Welch bound,'' in Proc.  IEEE Int. Symp. Inf. Theory, Vail, CO, USA, June 2018, pp. 2346-2349.

\bibitem{Luo3} G. J. Luo, X. W. Cao, ``Two constructions of asymptotically optimal codebooks,'' Crypt. Commun., 2018. [online].  Available:https://doi.org/10.1007/s12095-018-0331-4

\bibitem{field} R. Lidl and H. Niederreiter, Finite fields. Cambridge university press, 1997.

\bibitem{Apl.CDMA} J. Massey and T. Mittelholzer, ``Welch's bound and sequence sets for code-division multiple-access systems,'' Sequences II, Springer New York, pp. 63-78, 1999.

\bibitem{Apl.qc} J. Renes, R. Blume-Kohout, A. Scot, and C. Caves, ``Symmetric informationally complete quantum measurements,'' J. Math. Phys., vol. 45, no. 6, pp. 2171-2180, 2004.


\bibitem{F4}F. Rahimi, ``Covering graphs and equiangular tight frames,''
Ph.D. Thesis, University of Waterloo, Ontario, 2016 (available at
http://hdl.handle.net/10012/10793).


\bibitem{MWEB} D. Sarwate, ``Meeting the Welch bound with equality,'' New York, NY, USA:Springer-Verlag, 1999, pp. 63-79.


\bibitem{MWEB1} T. Strohmer and R. Heath, ``Grassmannian frames with applications to coding and communication,'' Appl. Comput. Harmon. Anal., vol. 14, no. 3, pp. 257-275, 2003.


\bibitem{Welch} L. Welch, ``Lower bounds on the maximum cross correlation of signals,'' IEEE Trans. Inform. Theory, vol. 20, no. 3, pp. 397-399, 1974.



\bibitem{Apl.cd3} P. Xia, S. Zhou, and G. Giannakis, ``Achieving the Welch bound with difference sets,'' IEEE Trans. Inform. Theory, vol. 51, no. 5, pp. 1900-1907, 2005.


\bibitem{NM3Yu} N. Yu, ``A construction of codebooks associated with binary
sequences,'' IEEE Trans. Inform. Theory, vol. 58, no. 8, pp. 5522-5533, 2012.


\bibitem{NM5} A. Zhang and K. Feng, ``Two classes of codebooks nearly meeting the
Welch bound,'' IEEE Trans. Inform. Theory, vol. 58, no. 4, pp. 2507-2511, 2012.


\bibitem{NM6} A. Zhang and K. Feng, ``Construction of cyclotomic codebooks nearly
meeting the Welch bound,'' Des. Codes Cryptogr., vol. 63, no. 2, pp. 209-224, 2013


\bibitem{zhou}Z. Zhou, X. Tang, ``New nearly optimal codebooks from relative
difference sets,'' Adv. Math. Commun., vol. 5, no. 3, pp. 521-527, 2011.


\end{thebibliography}
\end{document}